\documentclass[11pt]{article}

\usepackage{amsmath}
\usepackage{amsthm}
\usepackage{amssymb}
\usepackage{amsfonts}

\usepackage{dsfont}
\usepackage{mathrsfs}
\newcommand{\bx}{{\bf x}}
\newcommand{\by}{{\bf y}}

\newcommand{\bb}{{\bf b}}
\newcommand{\ba}{{\bf a}}

\newcommand{\B}{\{0,1\}}

\newcommand{\FF}{\mathbb{F}}
\newcommand{\F}{\mathbb{F}_p}

\newcommand{\bfone}{{\boldsymbol 1}}

\newcommand{\bfx}{{\boldsymbol x}}
\newcommand{\bfy}{{\boldsymbol y}}
\newcommand{\bfz}{{\boldsymbol z}}

\renewcommand{\Pr}{{\bf Pr}}

\newtheorem{theorem}{Theorem}

\newtheorem{lemma}[theorem]{Lemma}
\newtheorem{dfn}[theorem]{Definition}

\newtheorem{corollary}[theorem]{Corollary}

\renewcommand{\Pr}{{\bf Pr}}
\newcommand{\ignore}[1]{{ }}

\newcommand{\simpath}[2]{$#1$-SIMPLE $#2$-PATH}
\newcommand{\sipath}[1]{$#1$-SIMPLE PATH}
\newcommand{\spath}[1]{SIMPLE $#1$-PATH}
\newcommand{\hampath}{HAMILTONIAN-PATH}
\newcommand{\poly}{\mathrm{poly}}
\reversemarginpar

\title{\bf {On $r$-Simple $k$-Path}
 }

\author{Hasan Abasi\hspace{.3in}  Nader H. Bshouty\hspace{.3in}
Ariel Gabizon\thanks{The research leading to these results has received funding from the European Community's Seventh Framework Programme (FP7/2007-2013) under grant agreement number 257575.}\hspace{.3in}  Elad
Haramaty\thanks{This research was partially supported by the Israel Science Foundation (grant number 339/10)}\\
Department of Computer Science\\
 Technion, Haifa}

\begin{document}
\maketitle
\begin{abstract}
An $r$-simple $k$-path is a {path} in the graph of length $k$ that
passes through each vertex at most $r$ times. The \simpath{r}{k}
problem, given a graph $G$ as input, asks whether there exists an
$r$-simple $k$-path in $G$. We first show that this problem is
NP-Complete. We then show that there is a graph $G$ that contains
an $r$-simple $k$-path and no simple path of length greater than
$4\log k/\log r$. So this, in a sense, motivates this problem
especially when one's goal is to find a short path that visits many
vertices in the graph while bounding the number of visits at each
vertex.

We then give a randomized algorithm that runs in
time
$$\poly(n)\cdot 2^{O( k\cdot \log r/r)}$$ that solves the \simpath{r}{k}
on a graph with $n$ vertices with one-sided error. We also show
that a randomized algorithm with running time $\poly(n)\cdot
2^{(c/2)k/ r}$ with $c<1$ gives a randomized
algorithm with running time $\poly(n)\cdot 2^{cn}$ for the
Hamiltonian path problem in a directed graph - an outstanding open problem.
So in a sense our algorithm is optimal up to an $O(\log r)$
factor.
\end{abstract}

\section{Introduction}

Let $G$ be a directed graph on $n$ vertices. A path $\rho$ is
called {\it simple} if all the vertices in the path are distinct.
The \spath{k} problem, given $G$ as input, asks whether there
exists a simple path in $G$ of length $k$. This is a
generalization of the well known \hampath\; problem that asks
whether there is a simple path passing through \emph{all}
vertices, i.e., a simple path of length $n$ in $G$.  As \hampath\;
is NP-complete, we do not expect to find polynomial time
algorithms for \spath{k} for general $k$. Moreover, we do not even
expect to find good approximation algorithms for the corresponding
optimization problem: the \emph{longest path problem}, where we
ask what is the length of the longest simple path in $G$. This is
because Bj\"orklund et al. \cite{BHK04} showed that the longest
path problem cannot be approximated in polynomial time to within a
multiplicative factor of $n^{1-\epsilon}$, for any constant
$\epsilon>0$, unless P=NP. This motivated finding algorithms for
\spath{k} with running time whose dependence on $k$ is as small as
possible. The first result in this venue by  Monien \cite{M85}
achieved a running time of $k!\cdot \poly(n)$. Since then, there
has been extensive research on constructing algorithms for
\spath{k} running in time $f(k)\cdot\poly(n)$, for a function
$f(k)$ as small as possible~\cite{B93,AYZ95,KMRR06,CLSZ07,K08}.
The current state of the art is $2^k\cdot poly(n)$ by Williams
\cite{W09} for directed graphs and $O(1.657^k)\cdot \poly(n)$ by
Bj{\"o}rklund \cite{BHKK10} for undirected graphs.

\subsection{Our results}

 In this paper we look at a
further generalization of \spath{k} which we call \simpath{r}{k}.
In this problem instead of insisting on $\rho$ being a simple path, we allow $\rho$
to visit any vertex a fixed number of times.
We now formally define the problem \simpath{r}{k}.
\begin{dfn}\label{dfn:r_simp}
Fix integers $r\leq k$.
Let $G$ be a directed graph.
\begin{itemize}
\item We say a path $\rho$ in $G$ is \emph{$r$-simple}, if each
vertex of $G$ appears in $\rho$ at most $r$ times.  Obviously, $\rho$
is a simple path if and only if it is a $1$-simple path.
\item The \emph{\simpath{r}{k}} problem, given $G$ as input,
asks whether there exists an $r$-simple path in $G$ of length $k$.
\end{itemize}
\end{dfn}

At first, one may wonder whether for some fixed $r>1$, \simpath{r}{k}
always has a polynomial time algorithm.
We show this is unlikely by showing that for any $r$, for some $k$ \simpath{r}{k} is NP-complete.
See Theorem \ref{thm:simpathNPC} in Section \ref{sec:prelim_results} for a formal statement and proof of this.
Thus, as in the case of \spath{k}, one may ask what is the best dependency of the running time on $r$ and $k$ that can be
obtained in an algorithm for \simpath{r}{k}.

Our main result is
\begin{theorem}\label{thm:main}
Fix any integers $r,k$ with $2\leq r\leq k$. There is a randomized
algorithm running in time $$\poly(n)\cdot O\left(r^{ \frac{2k}{r}
+O(1)}\right) = \poly(n)\cdot 2^{O( k\cdot \log r/r)}$$ solving
\simpath{r}{k}
 on a graph with $n$ vertices  with
 one-sided error.
\end{theorem}
One may ask how far from optimal is the dependency on $k$ and $r$
in Theorem \ref{thm:main}. Theorem \ref{thm:simpathNPC} implies
that a running time of $\poly(n)\cdot 2^{o(k/r)}$ would give an
algorithm with running time $2^{o(n)}$ for \hampath. Moreover,
even a running time of $\poly(n)\cdot 2^{c\cdot k/r}$, for a small
enough constant $c<1/2$, would imply a better algorithm for
\hampath\; than those of \cite{W09,BHKK10} which are the best
currently known. So, in a sense our algorithm is optimal up to an
$O(\log r)$ factor. We find closing this $O(\log r)$ gap, e.g. by
a better reduction to \hampath, or a better algorithm for
\simpath{r}{k}, to be an interesting open problem.

\subsection{Finding a path with many distinct vertices}
We give more motivation for the \simpath{r}{k} problem. Suppose we
are in a situation where we wish to find a relatively short path
passing through many distinct vertices. Note that an $r$-simple
path of length $k$ must path through at least $k/r$ distinct
vertices. Thus, in case, for example, a $2$-simple path of length
$k$ exists, our algorithm for \simpath{2}{k} can be used to find a
path of length $k$ with at least $k/2$ distinct vertices in time
$\poly(n)\cdot 2^{k/2}$. One may ask how this would compare to the
number of distinct vertices returned by the algorithms for
\spath{k}. We show there can be a large gap. Specifically, for any
given $k$ , we show there is a graph $G$ where all simple paths
are of length less than $4\cdot \log k$, but $G$ contains a
$2$-simple path of length $k$. See Theorem \ref{Gap} for a precise
statement.

%More generally, suppose we are in a situation where we wish to find a relatively short path
%passing through at least $t$ distinct vertices. It is not hard to see that if such a path exists in $G$, there exists
%a $t$-simple path passing through $t$ distinct vertices.
%Running our algorithm for  \simpath{r}{k} for every $1\leq  r\leq  t$ and $k=r\cdot t$,
%we could find such a path in $G$ in time $\poly(n)\cdot 2^{k\cdot\log r/ r} \leq \poly(n)\cdot 2^{t\cdot \log t}$.
%\anote{Check that above statement is true - requires search to decision reduction}
%Moreover, if there exists such an $r$-simple path in $G$, we return a path of length at most
%$r\cdot t$, and in any case we return a path of length at most $t^2$.

\section{Overview of the proof of Theorem \ref{thm:main}}
We give an informal sketch of Theorem \ref{thm:main}.
We are given a directed graph $G$ on $n$ vertices, and integers $r\leq k$. We wish
to decide if $G$ contains an $r$-simple path of length $k$.
There are two main stages in our algorithm.
The first is to reduce the task to another one concerning multivariate polynomials.
This part, described below, is very similar to \cite{AB13}.
\paragraph{Reduction to a question about polynomials}
We want to associate our graph $G$ with a certain multivariate
polynomial $p_G$.

We associate with the $i$'th vertex a variable $x_i$.
The monomials of the polynomial will correspond to the paths of length $k$ in $G$.
So we have
\[p_G(\bx)=\sum_{i_1\to i_2\to \cdots \to i_{k} \in G}
x_{i_1}\cdots x_{i_{k}},\]
 where $i_1\to i_2\to \cdots \to i_{k} \in G$
means that $i_1, i_2, \cdots , i_{k}$ is a directed path in $G$.
An important issue is \emph{over what field $\mathbb{F}$ is $p_G$
defined?} A central part of the algorithm is indeed choosing the
appropriate field to work over. Another issue is how efficiently
$p_G$ can be evaluated? (Note that it potentially contains $n^k$
different monomials.) Williams shows in \cite{W09} that using the
adjacency matrix of $G$ it can be computed in $\poly(n)$-time. See
Section \ref{sec:path-to-poly}. For now, think of $p_G$ as defined
over $\mathbb{Q}$, i.e., having integer coefficients. It is easy
to see that $G$ contains an $r$-simple path of length $k$ if and
only if $p_G$ contains a monomial such that the individual degrees
of all variables are at most $r$. Let us call such a monomial an
\emph{$r$-monomial}. Thus our task is reduced to checking whether
a homogenous polynomial of degree $k$ contains an $r$-monomial.

\paragraph{Checking whether $p_G$ contains an $r$-monomial}

Let us assume in this overview for simplicity that $p=r+1$ is
prime. Let us view $p_G$ as a polynomial over $\F$. One problem
with doing this is that if we have $p$ directed paths of length
$k$ passing through the same vertices in different order, this
translates in $p_G$ to $p$ copies of the same monomial summing up
to $0$. To avoid this we need to look at a variant of $p_G$ that
contains auxiliary variables that prevent this cancelation. For
details on this issue see \cite{AB13} and
Section~\ref{sec:path-to-poly}. For this overview let us assume
this does not happen. Recall that we have the equality $a^p=a$ for
any $a\in \F$. Let us look at a monomial that \emph{is not} an
$r$-monomial, say $x_1^{r+1}\cdot x_2=x_1^p\cdot x_2$. The
equality mentioned implies this monomial is equivalent as a
function from $\F^n$ to $\F$ to the monomial $x_1\cdot x_2$. By
the same argument, any monomial that is \emph{not} an $r$-monomial
will be `equivalent' to one of smaller degree. More generally,
$p_G$ that is homogenous of degree $k$ over $\mathbb{Q}$ will be
equivalent to a polynomial of degree smaller than $k$ as a
function from $\F^n$ to $\F$ if and only if it \emph{does not}
contain an $r$-monomial. Thus, we have reduced our task to the
problem of \emph{low-degree testing}. In this context, this
problem is as follows: Given black-box access to a function
$f:\F^n\to \F$ of degree at most $k$, determine whether it has
degree exactly $k$ or less than $k$, using few queries to the
function. Here, for a function $f:\F^n\to\F$, by its degree we
mean the total degree of the lowest-degree polynomial $p\in
\F[x_1,\ldots,x_n]$ representing it. Haramati, Shpilka and Sudan
\cite{HSS} gave an optimal solution (in terms of the number of
queries) to this problem for prime fields. A little work is
required to determine the exact running time of the test of
\cite{HSS} (in addition to the bound on the number of queries to
$f$). See Section \ref{sec:ldtest} for details. For details on
dealing with the case that $r+1$ is not prime, see Section
\ref{sec:rmon}.

\section{Definitions and Preliminary Results}\label{sec:prelim_results}
In this section we give some definitions and preliminary results
that will be used throughout this paper.

Let $G(V,E)$ be a directed graph where $V$ is the set of vertices
and $E\subseteq V\times V$ the set of edges. We denote by $n=|V|$
the number of vertices in the graph and by $m=|E|$ the number of
edges in the graph. A $k$-{\it path} or a {\it path of length $k$}
is a sequence $\rho=v_1,\ldots,v_{k}$ such that $(v_i,v_{i+1})$
is an edge in $G$ for all $i=1,\ldots,k-1$. A {\it path} is a
$k$-path for some integer $k>0$. A path $\rho$ is called {\it
simple} if all the vertices in the path are distinct. We say that
a path $\rho$ in $G$ is \emph{$r$-simple}, if each vertex of $G$
appears in $\rho$ at most $r$ times.  Obviously, a simple path is
a $1$-simple path.

Given as input a directed graph $G$ on $n$ vertices, the
\simpath{r}{k} problem asks for a given $G$ whether it contains an
$r$-simple path of length $k$. When $r=1$ then the problem is
called \spath{k}. The \sipath{r} problem asks for a given $G$ and
integer $k$ whether $G$ contains an $r$-simple $k$-path of length
$k$. The problem SIMPLE PATH is \sipath{1}.

In this paper we
will study the above problems.

The following result gives a reduction from \simpath{r}{k} to
\spath{k}.

\begin{lemma} If \simpath{r}{k} can be solved in time $T(r,k,n,m)$ then
\simpath{sr}{k} can be solved in time $T(r,k,sn,s^2m)$. In particular,
If \spath{k} can be solved in time $T(k,n,m)$ then
\simpath{r}{k} can be solved in time $T(k,rn,r^2m)$.
\end{lemma}
\begin{proof} Let $G$ be a directed graph. Define the graph $G'=G\odot I_s$
where each vertex $v$ in $G$ is replaced with an independent set
$I_s$ of size $s$ in $G'$ with the vertices
$v^{(1)},\ldots,v^{(s)}$. Each edge $(u,v)$ in $G$ is replaced by
the edges $(u^{(i)},v^{(j)})$, $1\le i,j\le s$.

It is easy to see that there is a $rs$-simple $k$-path in $G$ if and only if there
is a $r$-simple $k$-path in $G'$.
\end{proof}

We now show that \sipath{r} is NP-complete.
\begin{theorem}\label{thm:simpathNPC} For any $r$ the decision
problem \sipath{r} is NP-complete.
\end{theorem}
\begin{proof}
We will reduce deciding \hampath\;  on a graph of $n$ vertices, to deciding \simpath{r}{(2rn-n+2)}
on a graph of $2\cdot n$ vertices.\\
Given an input graph $G=(V,E)$ to \hampath, we define a new graph $G'=(V',E')$ as follows.
We let $V' = V \bigcup \bar{V}$, where $\bar{V}=\{ \bar{v}_1,\bar{v}_2,...,\bar{v}_n\}$
and $E'=E\bigcup \bar{E}$ where
\[\bar{E}=\{(\bar{v}_i,v_i),(v_i,\bar{v}_i) | i\in[n] \}.\]
For $j\in [n]$, it will be convenient to denote by $\rho_j$, the
path of length $2r-1$ that begins at $v_j$, goes back and forth
from $v_j$ to $\bar{v}_j$ and ends in $v_j$,
 i.e., $\rho_j \triangleq (v_j,\bar{v}_j,\ldots,v_j,\bar{v}_j,v_j)$.

 We make 2 observations.\\
 \begin{enumerate}
 \item If a vertex $\bar{v}_j\in \bar{V}$ appears $r$ times in an $r$-simple path $\rho$ then it must be
 the start or end vertex of $\rho$: To see this, note that by construction of $G'$, if $\bar{v}_j$ is \emph{not}
 the start or end vertex of $\rho$, visiting it $r$ times requires visiting $v_j$ $r+1$ times.
\item Suppose $\rho$ is an $r$-simple path that begins and ends in a vertex of $V$.
If $\rho$ visits a vertex $\bar{v}_j\in \bar{V}$ $r-1$ times, then it must contain $\rho_j$ as a subpath:
To see this, note that as $\rho$ does not start in $\bar{v}_j$, any visit to $\bar{v}_j$ must have a visit to $v_j$ before and
after. The only way this would sum up to at most $r$ visits in $v_j$ is if these visits where continuous. In other words, only if $\rho$ contains $\rho_j$.
\end{enumerate}

We want to show that $G$ contains a Hamiltonian path if and only
if $G'$ contains an $r$-simple path of length $2rn-n+2$. Assume
first that $G$ contains a Hamiltonian path $v_{i_1}\cdot
v_{i_2}\cdots v_{i_n}$. Define the path $\rho =
\bar{v}_{i_1}\cdot\rho_{i_1}\cdot\rho_{i_2}\cdots\rho_{i_n}\cdot
\bar{v}_{i_n}$. It is of length
\[n\cdot(2r-1)+2 = 2rn - n +2, \]
and it is $r$-simple.\\
Now assume that we have an $r$-simple path $\rho$ in $G'$ of
length $2nr-n+2$. We first claim that $\rho$ must start and end
with a vertex from $\bar{V}$: Otherwise, using the first
observation above, $\rho$ contains at most $n+1$ vertices
appearing $r$ times, and thus has length at most
\[(n+1)\cdot r + (n-1)\cdot (r-1)  = 2rn - n+1 .\]
Let $\rho'$ be the path $\rho$ with the first and last vertex
deleted. So $\rho'$ has length $2rn - n$ and begins and ends in
a vertex of $V$. Note that by the first observation $\rho'$ visits
all vertices of $\bar{V}$ at most $r-1$ times. We now claim that
for every $j\in [n]$, $\rho'$ must contain $\rho_j$ as a subpath.
Otherwise, by the second observation, $\rho'$ visits \emph{some}
vertex of $\bar{V}$ \emph{less} than $r-1$ times. In this case it
has length less than $n\cdot r + n\cdot(r-1)  = 2nr -n$. A
contradiction. Thus $\rho'$ contains every $\rho_j$ as a subpath.
It cannot contain anything else `between' the $\rho_j$'s, as then
it would visit some vertex of $V$ \emph{more} than $r$ times. So
\[\rho'=\rho_{i_1}\cdots\rho_{i_n},\]
for some ordering $i_1,\ldots,i_n$ of $[n]$. It follows that $v_{i_1}\cdots v_{i_n}$ is a Hamiltonian path in $G$.
\end{proof}

The above result implies
\begin{corollary}\label{LoB} If \simpath{r}{k} can be solved in $T(r,k,n,m)$
time then {\rm HAMILTONIAN-PATH} can be solved in
$T(r,2rn-n+2,2n,m+2n)$.

In particular, if there is an algorithm for \simpath{r}{k} that
runs in time $poly(n)\cdot 2^{(c/2)(k/r)}$ then there is an
algorithm for {\rm HAMILTONIAN-PATH} that runs in time
$poly(n)\cdot 2^{cn}$.
\end{corollary}

\section{Gap}

In this section we show that the gap between the longest simple
path and the longest $r$-simple path can be exponentially large
even for $r=2$.

We first give the following lower bound for the gap

\begin{theorem} If $G$ contains an $r$-simple path of length
$k$ then $G$ contains a simple path of length $\lceil \frac{\log
k}{\log r}\rceil$.
\end{theorem}
\begin{proof}
Let $t=\lfloor \frac{\log {k} -1}{\log r}\rfloor$.
Let $\rho$ be an $r$-simple path whose first vertex is $v_0$. We will use $\rho$ to construct a simple path $\bar{\rho}$ of length $\lfloor\frac{logk}{logr}\rfloor$.
We denote $\rho_0=\rho$. As $v_0$ appears at most $r$ times in $\rho_0$, there must
be a subpath $\rho_1$ of $\rho_0$ of length at least $(k-r)/r$ where $v_0$ does not appear.
Let $v_1$ be the first vertex of $\rho_1$.
Similarly, for $1< i \leq t$, we define the subpath $\rho_i$ of $\rho_{i-1}$ to be
a subpath of length at least
\[(k-r-\ldots -r^i)/r^i \geq (k-r^{i+1})/r^i,\]
where $v_1,\ldots,v_{i-1}$ do not appear, and define $v_i$ to be the first vertex of $\rho_i$.
Note that we can always assume there is an edge from $v_{i-1}$ to $v_i$ as we can start $\rho_i$ just
after an appearance of $v_{i-1}$ in $\rho_{i-1}$.
Note that for $1\leq i \leq t$, such a $v_i$ as defined indeed exists as
$(k-r^{i+1})/r^i\geq 1$ when
\[ k\geq 2\cdot r^{i+1}\leftrightarrow i+1 \leq (\log k -1)/\log r\]
Thus, $v_0\cdots v_{t-1}$ is a simple path of the desired length.
\end{proof}
%\begin{theorem} If $G$ contains an $r$-simple path of length
%$k$ then $G$ contains a simple path of length $\lceil \frac{\log
%k}{\log r}\rceil-2$.
%\end{theorem}
%\ignore{\begin{proof}
%
%We call a path $p=v_1,v_2,\ldots,v_{k_1+k_2+1}$ in $G$, $r$-simple
%$(k_1,k_2)$-path if $v_1,\ldots,v_{k_1+1}$ is a simple path and
%$v_{k_1+1},\ldots,v_{k_1+k_2+1}$ is $r$-simple. We show by
%induction on $k_1$ that for every $k_1\le \lceil {(\log k)}/{(\log
%r)}\rceil-1$, $G$ contains an $r$-simple $(k_1,k_2)$-path that
%satisfies
%$$r^{k_1}(k_2+4)\ge k+4.$$
%
%Let $p$ be $r$-simple path of length $k$ in $G$. Then $p$ is
%$r$-simple $(0,k)$-path. Since $r^0(k+4)\ge k+4$, the basis of the
%induction is true. Now suppose $G$ contains an $r$-simple
%$(k_1,k_2)$-path $p=v_1,v_2,\ldots,v_{k_1+k_2+1}$ that satisfies
%$r^{k_1}(k_2+4)\ge k+4.$ If $v_{k_1+1}$ appears $r'$ times in the
%path $v_{k_1+1},\ldots,v_{k_1+k_2+1}$ then by the pigeonhole
%principle there is a sub-path $p'=v_{i_1},\ldots,v_{i_1+\ell}$
%where $v_{i_1}=v_{k_1+1}$, $\ell\ge \lceil k_2/r'\rceil-1$ and
%$v_{k_1+1}$ appears once in $p'$. Then
%$p=v_1,v_2,\ldots,v_{k_1},v_{k_1+1},v_{i_1+1},\ldots,v_{i_1+\ell}$
%is an $r$-simple $(k_1+1, \ell-1)$-path. Now since
%$$r^{k_1+1}((\ell-1)+4)\ge
%r^{k_1+1}(\lceil k_2/r'\rceil+2)\ge r^{k_1}(k_2+4)\ge k+4$$ the
%result follows.
%
%Notice that an $r$-simple $(k',0)$-path satisfies $r^{k'}4\ge k+4$
%and therefore $$k'\ge \left\lceil \frac{\log(k/4+1)}{\log
%r}\right\rceil \ge \left\lceil \frac{\log k}{\log
%r}\right\rceil-2.$$
%\end{proof}

Before we give the upper bound we give the following definition. A
{\it full $r$-tree} is a tree where each vertex has $r$ children
and all the leaves of the tree are in the same level. The root is
on level $1$.
\begin{theorem}\label{Gap} There is a graph $G$ that contains an $r$-simple
path of length $k$ and no simple path of length greater than
$4\log k/\log r$.
\end{theorem}
\begin{proof} We first give the proof for
$r\ge 3$. Consider a full $(r-1)$-tree of depth $\lceil \log
n/\log(r-1)\rceil$. Remove vertices from the lowest level (leaves)
so the number of vertices in the graph is $n$. Obviously there is
an $r$-simple path of length $k\ge n$. Any simple tour in this tree
must change level at each step and if it changes from level $\ell$
to level $\ell+1$ it cannot go back in the following step to level $\ell$. So the
longest possible simple path is $2\lceil \log
n/\log(r-1)\rceil-2\le 3.17 (\log k/\log r)$.

For $r=2$ we take a full binary tree ($2$-tree) and add an edge
between every two children of the same vertex. The $2$-simple path
starts from the root $v$, recursively makes a tour in the left tree
of $v$ then moves to the root of the right tree of $v$ (via the
edge that we added) then recursively makes a tour in the right tree
of $v$ and then visit $v$ again. Obviously this is a $2$-simple
path of length $k>n$. A simple tour in this graph can stay in
the same level only twice, can move to a higher level or can move
to a lower level. Again here if it moves from level $\ell$ to
$\ell+1$ it cannot go back in the following step to level $\ell$. Therefore the longest
simple path is of length at most $4\log n \leq 4\log k$.
\end{proof}

\ignore{
\begin{theorem}\label{thm:gap}
If $G$ contains an $r$-simple path of length $k$ then there is $G$ which contains no simple path of length  $\geq 4\lfloor \frac{\log k}{\log r}\rfloor$.
\end{theorem}

\begin{corollary} The Gap is $\lfloor \frac{\log k}{\log r}\rfloor$ .
\end{corollary}}

\section{From $r$-Simple $k$-Path to Multivariate Polynomial}\label{sec:path-to-poly}
The purpose of this section is to reduce the question of whether a graph $G$ contains
an $r$-simple $k$-path, to that of whether a certain multivariate polynomial \emph{contains an $r$-monomial},
as defined below.
\begin{dfn}[$r$-monomial]
Fix a field $\FF$.
Fix a monomial  $M= M(\bfz)=z_1^{i_1}\cdots z_t^{i_t}$.
\begin{itemize}
\item
We say $M$ is an \emph{$r$-monomial} if no variable appears with degree larger than $r$ in $M$.
That is, for all $1\leq j \leq t$, $i_j\leq r$.
\item  Let $f(\bfz)$ be a multivariate polynomial over $\FF$.
We say $f$ \emph{contains an $r$-monomial}, if there is an $r$-monomial $M(\bfz)$ appearing with a nonzero coefficient
$c\in \FF$ in $f$.
\end{itemize}
\end{dfn}

We now describe this reduction.

Let $G(V,E)$ be a directed graph where $V=\{1,2,\ldots,n\}$. Let
$A$ be the adjacency matrix and $B$ be the $n\times n$ matrix such
that $B_{i,j}=x_i\cdot A_{i,j}$ where $x_i$, $i=1,\ldots,n$ are
indeterminates. Let $\bf1$ be the row $n$-vector of $1$s and
$\bfx=(x_1,\ldots,x_n)^T$. Consider the polynomial
$p_G(\bfx)=\bfone \cdot B^{k-1}\cdot \bfx$. It is easy to see
$$p_G(\bfx)=\sum_{i_1\to i_2\to \cdots \to i_{k} \in G}
x_{i_1}\cdots x_{i_{k}}$$ where $i_1\to i_2\to \cdots \to i_{k} \in G$
means that $i_1, i_2, \cdots , i_{k}$ is a directed path in $G$.

Obviously, for field of characteristic zero there is an $r$-simple
$k$-path if and only if $p_G(\bfx)$ contains an $r$-monomial.
For other fields the later statement is
not true. For example, in undirected graph, $k=2$, and $r=1$ if
$(1,2)\in E$ and the field is of characteristic $2$ then the
monomial $x_1x_2$ occurs twice and will vanish in $p_G(\bfx)$. We
solve the problem as follows.

Let $B^{(m)}$ be an $n\times n$ matrices, $m=2,\ldots,k$, such
that $B^{(m)}_{i,j}=x_i\cdot y_{m,i}\cdot A_{i,j}$ where $x_i$ and
$y_{m,i}$ are indeterminates. Let, $\bfy=(\bfy_1,\ldots,\bfy_{k})$
and $\bfy_m =(y_{m,1},\ldots,y_{m,n})$. Let
$\bfx\centerdot\bfy=(x_1y_{1,1},\ldots,x_ny_{1,n})$. Consider the
polynomial $P_G(\bfx,\bfy)=\bfone B^{(k)}B^{(k-1)}\cdots
B^{(2)}(\bfx\centerdot\bfy)$. It is easy to see that

$$P_G(\bfx,\bfy)=\sum_{i_1\to i_2\to \cdots \to i_{k} \in G} x_{i_1}\cdots x_{i_{k}} y_{1,i_1}\cdots y_{k,i_{k}}$$
Obviously, no two paths have the same monomial in $P_G$.
Note that as $P_G$ contains only $\B$ coefficients, we can define it over any field $\FF$.
It will actually be convenient to view it as a polynomial $P_G(\bfx)$ whose coefficients are in the field
of rational functions $\FF(\bfy)$.
Therefore, for any field, there is an $r$-simple $k$-path if and
only if $P_G(\bfx,\bfy)$ contains an $r$-monomial in $\bfx$.
We record this fact in the lemma below.
\begin{lemma}\label{lem:path_to_pG}
Fix any field $\FF$.
The graph $G$ contains an $r$-simple $k$-path if and only if
the polynomial $P_G$, defined over $\FF(\bfy)$, contains an $r$-monomial $M(\bfx)$.
\end{lemma}

\section{Low Degree Tester}\label{sec:ldtest}
In this section we present a tester that determines whether a
function $f:\F^n\to \F$ of degree \emph{at most} $d$ has, in fact,
degree \emph{less than} $d$. The important point is that the
tester will be able to do this using few black-box queries to $f$.
The results of this section essentially follow from the work of
Haramaty, Sudan and Shpilka \cite{HSS}.

First, let us say precisely what we mean by the \emph{degree} of a function $f:\F^n\to \F$.
We define this to be the degree of the lowest degree polynomial $f'\in\F[\bx]$ that agrees with
$f$ as a function from $\F^n$ to $\F$.
It is known from the theory of finite fields that there is a unique such $f'$, and
that the individual degrees of all variables in $f'$ are smaller than $p$.
Moreover, given any polynomial $g\in\F[\bfx]$ agreeing with $f$ as a function from $\F^n$ to $\F$,
$f'$ can be derived from $g$ by replacing, for any $1\leq i \leq n$, occurrences of $x_i^t$
with $x_i^t \mod x_i^p-x_i$ (i.e., $x_i^{((t-1)\mod (p-1))+1}$ when $t\not=0$).
We do not prove these basic facts formally here. They essentially follow from the fact that
$a^p=a$ for $a\in \F$.

This motivates the following definition.
\begin{dfn}\label{dfn:degp}
Fix positive integers $n,d$ and a prime $p$.
Let $f\in\F[\bfx] = \F[x_1,\ldots,x_n]$.
 We define $\deg_p(f)$ to be the degree of the polynomial $f$ when
replacing, for $1\leq i \leq n$, $x_i^t$ by $(x_i^t\mod x_i^p-x_i)$. More formally,
$\deg_p(f)\triangleq \deg(f')$ where
\[f'(x_1,\ldots,x_n) \triangleq f(x_1,\ldots,x_n)\; \mathrm{mod} \;x_1^p-x_1,\;\ldots,\;\mathrm{mod}\; x_n^p-x_n.\]
Moreover, for a function $g:V \to \FF_p$ where $V \subseteq
\FF^n_p$ is a subspace of dimension $k$, we define
$\deg_p(g)=\min_f deg_p(f)$ where $f\in \FF_p [x_1,...,x_n]$ and
$f|_V=g$. Here $g$ can be regarded as a function in
$\FF_p[x_1,\ldots,x_k]$.
\end{dfn}

We note that this notion of degree is affine invariant, i.e does
change after affine transformations. In addition it has the
property that for any affine subspace $V$, $\deg_p(f|_V) \leq
\deg_p(f)$.

We now present the main result of this section.

\begin{lemma}\label{lem: low degree testing}
There is a randomized algorithm $A$ running in time
$\poly(n)\cdot p^{\left\lceil \frac{d}{p-1}\right\rceil +1}$
that determines with constant one-sided error whether a function $f$ of degree
at most $d$ has degree less than $d$.
More precisely, given black-box access to a function $f:\F^n\to \F$ with  $\deg_p(f)\leq d$,

\begin{itemize}
\item If $\deg_p(f)= d$, $A$ accepts with probability at least $99/100$.
\item If $\deg_p(f)<d$, $A$ rejects with probability one.
\end{itemize}
\end{lemma}

The result essentially follows from the work of Haramaty, Shpilka and Sudan \cite{HSS}.
A technicality is to analyze the precise running time, and not just the query complexity as in \cite{HSS}

Before proving Lemma \ref{lem: low degree testing}, we state some required preliminary lemmas.

\begin{lemma}\label{lem: not-so-low degree testing}
Suppose we have black-box access to a function $f:\F^t\to \F$.
Then we can determine in deterministic time $O(p^t)$ whether $deg_p(f)\geq (p-1)\cdot t$.
\end{lemma}

\begin{proof}
Consider the algorithm that yields a positive answer if and only if
$\sum_{\ba\in\FF_{p}^{t}}f(\ba)=0$ .
It is clear that the running time is indeed $O(p^{t})$.
Let us now show correctness.
As in Definition \ref{dfn:degp}, define
\[f'(x_1,\ldots,x_n) \triangleq f(x_1,\ldots,x_n)\; \mathrm{mod} \;x_1^p-x_1,\;\ldots,\;\mathrm{mod}\; x_n^p-x_n,\]
so that $\deg_p(f) = \deg(f')$.
We show that

\begin{enumerate}
\item The only monomial of degree
$\geq t(p-1)$ in $f'$  is $M_{\max}\triangleq \prod_{i=1}^{t}x_{i}^{p-1}$ and
\item  the coefficient of $M_{\max}$ in $f'$, is  $(-1)^{t}\cdot \sum_{\ba\in\F^t}f(\ba)$.
\end{enumerate}
~From these two items, it is clear that indeed $\deg_p(f)=\deg(f')
\geq t\cdot (p-1)$ if and only if $ \sum_{\ba\in\F^t}f(\ba)\neq
0$.

The first item is obvious, as the individual degrees in $f'$ are at most $p-1$.

 For the second item, let us calculate the coefficient of $M_{\max}$ in $f'$.
 For every $\ba\in \F^t$, consider the
function $g_{\ba}:\F^t\to \F$ that is one on $\ba$ and zero elsewhere. One can
verify that $g_{\ba}(\bx)=\prod_{i=1}^{t}\frac{\prod_{\alpha\in\F\backslash\{\ba_i\}}(\bx_{i}-\alpha)}{\prod_{\beta\in\F\backslash\{0\}}\beta}$.
Clearly, the coefficient of $M_{\max}$ in $g_{\ba}$ is $(\prod_{\beta\in\F\backslash\{0\}}\beta)^{-t}=(-1)^t$.
Note that in $g_{\ba}$, all individual degrees are smaller than $p$.
Hence, $f'=\sum_{\ba\in\F^{t}}f(\ba)\cdot g_{\ba}$ and the coefficient of
$M_{\max}$ in $f'$ is $(-1)^{t}\cdot \sum_{\ba \in \F^t} f(\ba)$.
\end{proof}

The algorithm for Lemma~\ref{lem: low degree testing} checks the
degree of the function only on a small subspace. The key for its
correctness is to show that when you restrict the function to a
subspace (even for $n-1$ dimensional subspace) the degree does not
decrease with high probability. The Lemma appeared in \cite{HSS}.
We give a proof sketch here for completeness

\ignore{ We state two
lemmas from \cite{HSS}.
\begin{lemma}
Let $f$ be a degree $t(p-1)$ polynomial. Then there is an
affine transformation $A$ s.t $f\circ A$ contain a monomial $\prod_{i=1}^{t}x_{i}^{p-1}$.
\end{lemma}

\begin{proof}
 For two different degree $t(p-1)$ monomials $M_{1}=\prod_{i=1}^{n}x_{i}^{a_{i}}$
and $M_{2}=\prod_{i=1}^{n}x_{i}^{b_{i}}$ we say that $M_{1}$ is
lexicographic greater then $M_{2}$ if for the minimal index $k$
such that $a_{k}\neq b_{k}$ it is the case where $a_{k}>b_{k}$.
Clearly, the maximal monomial among them is $M_{\max}\triangleq\prod_{i=1}^{t}x_{i}^{p-1}$.
We will denote by $M_{f}$ the maximal monomial of $f$ of degree
$t(p-1)$ according to this order. We now show that if $M_{f}\neq M_{\max}$
then there is an affine transformation $A_{0}$ such that $M_{f\circ A_{0}}$
is lexicographic higher then $M_{f}$. By repeating this argument
one can verify that there is an affine transformation $A=A_{0}A_{1}....A_{m}$
such that $M_{f\circ A}=M_{\max}$.

Assume that $M_{f}\neq M_{\max}$ and let $k\in[t]$ be the first
index satisfy $a_{k}\neq p-1$. because the total degree is $t(p-1)$
then there is an index $m>t$ such that $a_{m}>0$. We show the existing
of such an $A_{0}$ by showing that there is an $\alpha\in\FF_{p}$
s.t $f\left(x_{1},...,x_{m}+\alpha x_{k},...,x_{n}\right)$ contain
a monomial that is lexicographic grater then $M_{f}=\prod_{i=1}^{n}x_{i}^{a_{i}}$.

Consider the function $f'(x_{1},...,x_{n},z)=f\left(x_{1},...,x_{m}+zx_{k},...,x_{n}\right)$.
Observe:
\[
M_{f}(x_{1},...,x_{m}+zx_{k},...,x_{n})=\sum_{b=0}^{a_{m}}{a_{m} \choose b}\prod_{i\neq k,m}^{n}x_{i}^{a_{i}}x_{k}^{a_{k}+b}x_{m}^{a_{m}-b}z^{b}\;.
\]
Moreover the monomial $\prod_{i\neq k,m}^{n}x_{i}^{a_{i}}x_{k}^{a_{k}+b}x_{m}^{a_{m}-b}z^{b}$
in $f'$ can generated only from $M_{f}$ and no other monomial of
$f$. so they cannot be canceled in $f'$. Specifically, for $b=1$
we have the monomial $\prod_{i\neq k,m}^{n}x_{i}^{a_{i}}x_{k}^{a_{k}+1}x_{m}^{a_{m}-1}z$
in $f$ (is coefficient multiply by ${a_{m} \choose 1}=a_{m}\not\equiv_{p}0$.

Now consider $f'$ as polynomial over $x_{1},...,x_{n}$ with coefficient
in $\FF_{p}[z]$. We saw that the coefficient of the monomial $\prod_{i\neq k,m}^{n}x_{i}^{a_{i}}x_{k}^{a_{k}+1}x_{m}^{a_{m}-1}$
is depend on $z$ and hence not the zero polynomial. So there is some
$\alpha\in\FF_{p}$ such that it is not zero. For this alpha the function
$f'(x_{1},...,x_{n},\alpha)=f(x_{1},...,x_{m}+\alpha x_{k},...,x_{n})$
has the monomial $\prod_{i\neq k,m}^{n}x_{i}^{a_{i}}x_{k}^{a_{k}+1}x_{m}^{a_{m}-1}$
which is lexicographic greater then $M_{f}$ and we are done.
\end{proof}
}
% Preview source code from paragraph 46 to 50
\begin{lemma}[Theorem 1.5 in \cite{HSS}]\label{lem:HSS-step}
Let $\FF_{p}$ be a field of prime size $p$ and $f:\F^n\to\F$ be a
function with  $\deg_p(f) = t(p-1)$. The number of
hyperplanes $H$ such that $\deg_p(f|_{H})<t(p-1)$ is at most
$p^{t+1}$\end{lemma}
\begin{proof}[Proof sketch]
We will assume w.l.o.g that $f$ has the monomial
$\prod_{i=1}^{t}x_{i}^{p-1}$. One can show that for any degree
$t(p-1)$ polynomial $f$ there is linear transformation $A$ such
that $f(Ax)$ has the monomial $\prod_{i=1}^{t}x_{i}^{p-1}$. So it
will be enough to prove the lemma for the suitable transformation
of $f$.

We will assume for simplicity that all the hyperplanes are of the
form of $H_{\alpha}=\{x\mid
x_{1}=\sum_{i=2}^{n}\alpha_{i}x_{i}+\alpha_{0}\}$ for some
$\alpha_{2},\ldots,\alpha_{n}$. Indeed, there are few more
hyperplanes that does not depend on the first coordinate, but they
don't contribute much to the upper bound.

To prove the lemma we will show that for any of the $p^{t}$
possible values for $\alpha_{2},\ldots,\alpha_{t},\alpha_{0}$
there are $<p$ possibilities for $\alpha_{t+1},\ldots,\alpha_{n}$
such that $\deg(f|_{H_{\alpha}})<t(p-1)$. Fix
$\alpha_{2},\ldots,\alpha_{t},\alpha_{0}$. For simplicity we
assume they are all zero, but the same bound goes for any
$\alpha_{2},\ldots,\alpha_{t},\alpha_{0}$ (one can reduce the
general case to the zero case by some affine transformation).

Now consider all the monomials $M$ in $f$ with the following
properties: (1) $M$ divides $\prod_{i=2}^{t}x_{i}^{p-1}$ and (2)
$\deg(M)=t(p-1)$. We can write the sum of all those monomials as $
$$\prod_{i=2}^{t}x_{i}^{p-1}g(x_{1},x_{t+1},\ldots,x_{n})$. By
definition, $g$ is homogenous polynomial of degree $p-1$. Because
$\prod_{i=1}^{t}x_{i}^{p-1}$ is a monomial of $f$, $x_{1}^{p-1}$
is a monomial of $g$.

Because the hyperplanes does not depend on the variables
$x_{2},\ldots,x_{t}$ (recall, we assumed $\alpha_2 = \cdots =
\alpha_t = 0$) the degree of $f$ can decrease on $H_{\alpha}$ only
if the degree of $g$ decrees on $H_{\alpha}$. Because $g$ is
homogenous of degree $p-1$ and we consider only linear hyperplanes
of the form $x_{1}=L(x_{t+1},\ldots,x_{n})$, then $g|_{x_{1}=L}$
is still homogenous of degree $p-1$, so if the degree
$\deg(g|_{x_{1}=L})<p-1$ then $g|_{x_{1}=L}\equiv0$. Now consider
$g$ as an univariate polynomial in $x_{1}$ over the field of
rational functions in $x_{t+1},\ldots,x_{n}$. In this view our
question is: how many field elements
$L\in\FF_{p}(x_{t+1},\ldots,x_{n})$ are there such that $g(L)=0$.
From the fundamental theorem of the algebra  the answer is $p-1$
and we are done.
\end{proof}

~From Lemma~\ref{lem:HSS-step} we get the following corollary.
\begin{corollary}\label{cor: low deg tester: completeness}
Let $n>t$ and $f:\FF_p^n \to \FF_p$  be a polynomial such that
$\deg_p(f)=t(p-1)$. Then
$\Pr_{V}\left[\deg_p(f|_{V})=t(p-1)\right]\geq\frac{1}{p+1}\prod_{k=1}^{n-t-1}\left(1-p^{-k}\right)=\Omega(\frac{1}{p})$,
where $V$ is a random $t$-dimensional affine subspace.
\end{corollary}

\begin{proof}
We proceed by induction on $n$. Consider first the base case,
where $n=t+1$. In this case the number of $t$-dimensional affine
subspaces $V \subseteq \FF^{t+1}_p$ is $\frac{p^{t+2} -1}{p-1} >
p^{t+1} + p^{t}$. By Lemma~\ref{lem:HSS-step} on at most $p^{t+1}$
of them $\deg(f|_V) <t(p-1)$ so the probability that $\deg(f|_V) =
t(p-1)$ is $\frac{1}{p+1}$ as
 required.

Now assume the claim is true for $n-1$, and consider the following
way of choosing a random $t$-dimensional affine subspace $V$.
First choose a random hyperplane $H\subseteq \FF_p^n$ and then
choose a random $t$-dimensional affine subspace $V\subseteq H$.
There are more than $p^n$ hyperplanes $H\subseteq \FF_p^n$, so by
Lemma~\ref{lem:HSS-step} the probability that
$\deg_p(f|_H)=t(p-1)$ is at least  $1-p^{t+1-n}$. Moreover, in the
event that $\deg_p(f|_H)=t(p-1)$, we can apply the induction
hypothesis to $f|_H$. Hence,
\begin{eqnarray*}
\Pr\left[\deg_p(f|_{V})=t(p-1)\right] & = & \Pr\left[\deg_p((f|_{H})|_{V})=t(p-1)\mid\deg_p(f|_{H})=t(p-1)\right]\cdot\Pr\left[\deg_p(f|_{H})=t(p-1)\right]\\
 & \leq & \frac{1}{p+1}\prod_{k=1}^{(n-1)-t-1}\left(1-p^{-k}\right)\cdot(1-p^{t+1-n})=\frac{1}{p+1}\prod_{k=1}^{n-t-1}\left(1-p^{-k}\right)
\end{eqnarray*}
\end{proof}

We are now ready to prove Lemma~\ref{lem: low degree testing}.
\begin{proof}[Proof of Lemma~\ref{lem: low degree testing}]
Let $t= \left\lceil \frac{d}{p-1}\right\rceil$. We assume without
lost of generality that $d = t(p-1)$: Otherwise, let $a=t(p-1)-d$
and consider the function $f'(x_0,x_1,....,x_n) \triangleq
x_0^a\cdot  f(x_1,...,x_n)$. It is easily checked that
$\deg_p(f')\leq t(p-1)$. Also $\deg_p(f')=t(p-1)$ if and only if
$\deg_p(f) =d $.

 We will present an algorithm for the problem with
one sided error probability $1-\Omega\left(\frac{1}{p}\right)$
that runs in time $\poly(n)\cdot O(p^t)$. By repeating it $O(p)$
times, we can get down to error probability $1/100$ in running
time $\poly(n)\cdot O(p^{t+1})$ as required.

Consider the following algorithm. Choose a random $t$-dimensional affine subspace $V$. Accept if and only if the $\deg_p(f|_V)<t(p-1)$.
Assume first that $\deg_p(f)<t(p-1)$.
Then for any affine subspace $V$, $\deg_p(f|_V) \leq \deg_p(f)<t(p-1)$.
On the other hand, if $\deg_p(f|_V)= t(p-1)$, Corollary~\ref{cor: low deg tester: completeness} implies we will accept with probability at least $\Omega(\frac{1}{p})$.

We conclude by analyzing the running time.
Choosing $V$ can be done in $\poly(n)$-time.
For checking whether $\deg_p(f|_V) =t(p-1)$, Lemma~\ref{lem: not-so-low degree testing} gives running
$O(p^t)$ assuming black-box access to $f|_V$.
Given black-box access to $f$, we can compute
$f|_V(\ba)$ for $\ba\in \F^t$ in $\poly(n)$-time.
The claimed running time of $\poly(n)\cdot O(p^t)$ follows.
\end{proof}

\section{Testing if $P_G$ contains an $r$-monomial}\label{sec:rmon}

In this section we present a method for testing
whether the polynomial $P_G$, described in Section \ref{sec:path-to-poly}, contains an $r$-monomial.
 This is done using the low-degree tester from the previous section.

As stated in Lemma \ref{lem:path_to_pG}, this is precisely equivalent to whether $G$ contains an $r$-simple $k$-path.
Recall we viewed $P_G$ as a polynomial over a field of rational functions $\F(\by)$.
To obtain efficient algorithms, we first reduce the question
to checking whether a different polynomial defined over $\F$ rather than $\F(\by)$
contains an $r$-monomial.
It is important in the next Lemma that we are able to do this reduction for \emph{any} $p$, in particular a
 `small' one.

\begin{lemma}\label{lem:only_x_var}
Fix any integers $r,k$, with $r\leq k$. Let $p$ be any prime and
$t=\lceil \log_p {10k} \rceil $.  Let $G$ be a directed graph on
$n$ vertices. Given an adjacency matrix $A_G$ for $G$, we can
return in $\poly(n)$-time $\poly(n)$-size circuits computing
polynomials $f^1_G,\ldots,f^t_G:\F^n\to \F$ on inputs in $\F^n$
such that
\begin{itemize}
\item For $1\leq i \leq t$, $f^i_G$ is (either the zero polynomial or) homogenous of degree $k$.
\item If $G$ contains an $r$-simple $k$-path then with
probability at least $9/10$, for some $1\leq i\leq t$, $f^i_G$ contains an $r$-monomial.
\item If $G$ does not contain an $r$-simple $k$-path, for all $1\leq i \leq t$, $f^i_G$
does not contain an $r$-monomial.
\end{itemize}
\end{lemma}
\begin{proof}
Note that the discussion in Section \ref{sec:path-to-poly} implies
we can compute $P_G$ in  $\poly(n)$-time over inputs in
$\FF_{p}^{2n}$. We choose random $\bb\in \FF_{p^t}^n$ and let
\[f_G(\bx)\triangleq P_G(\bx,\bb).\]
Suppose $P_G$, as a polynomial over $\FF(\by)$, contains an $r$-monomial $M'(\bx)$.
The coefficient $c_{M'}(\bfy)$ of $M'$ in $P_G$ is a nonzero polynomial of degree $k$.
So, by the Schwartz-Zippel Lemma, $c_{M'}(\bb)=0$ with probability at most $k/p^t\leq 1/10$.
In the event that $c_{M'}(\bb)\neq 0$, $f_G(\bx)$ is a homogenous polynomial of degree $k$ in $\FF_{p^t}[\bx]$
containing an $r$-monomial. Let us assume from now on, we chose a $\bb$ such that
indeed $a_{M'}\triangleq c_{M'}(\bb)\neq 0$.
We now discuss how to end up with polynomials having coefficients in $\F$ rather than $\FF_{p^t}$.

Let $T_1,\ldots, T_t:\FF_{p^t}\to \F$ be independent $\F$-linear maps.
Suppose $f_G = \sum_M a_M\cdot M(\bx)$.
For $1\leq i \leq t$, define a polynomial $f_G^i\in \F[\bx]$ by
\[f_G^i (\bfx) \triangleq \sum_M T_i(a_M)\cdot M(\bfx).\]
Note that for all $1\leq i \leq t$, $f_G^i$ is the zero polynomial
or homogenous of degree $k$. As $a_{M'}\neq 0$, for some $i$,
$T_i(a_{M'})\neq 0$. For this $i$, $f_G^i$ is homogenous of degree
$k$ and contains an $r$-monomial, specifically, the $r$-monomial
$a_{M'}\cdot M'(\bx)$. We claim that for all $1\leq i \leq t$,
$f_G^i$ can be computed by a $poly(n)$-size circuit on inputs
$\ba\in \F^n$. This is because $f_G$ and $T_i$ are efficiently
computable, and because for $\ba \in \F^n$,

\[T_i(f_G(\ba)) = T_i\left(\sum_M a_M\cdot M(\ba)\right)
= \sum_M T_i(a_M)\cdot M(\ba) = f_G^i(\ba),\] where the second
equality is due to the $\F$-linearity of $T_i$.

\end{proof}

The above lemma implies

\begin{corollary}\label{cor:enough_over_Fp}
Fix any prime $p$. Suppose that given black-box access to a
polynomial $g\in \F[\bx]$ that is homogenous of degree $k$, we can
determine in time $\poly(n)\cdot S$ if it contains an
$r$-monomial. Then we can also determine in time $\poly(n)\cdot S$
whether $P_G$ as a polynomial over $\F(\by)$ contains an $r$
monomial.
\end{corollary}

{Our reduction to low-degree testing is based on the following
simple observation that, for the right $p$ and for homogenous
polynomials, containing an $r$-monomial is equivalent to having a
certain $\deg_p$-degree.}
\begin{lemma}\label{lem:degiffr-mon}
Suppose $g\in\F[\bx]$ is a homogenous polynomial of degree $k$.
Suppose $r=p-1$. Then $\deg_p(g)=  k $ if and only if $g$ contains
an $r$-monomial.
\end{lemma}
\begin{proof}
If $g$ contains an $r$-monomial $M$ then, as $r<p$, $\deg_p(M)=k$,
which implies that $\deg_p(g)= k$. If $g$ does not contain an
$r$-monomial, then for every monomial $M$ in $g$ there is an $i\in
[n]$ such that the degree of $x_i$ in $M$ is at least  $r+1=p$. So
replacing $x_i^p$ by $x_i$ will reduce the degree of $M$ and
therefore $\deg_p(M)<k$. Since this happens for all monomials of
$g$, $\deg_p(g)< k$.
\end{proof}

We introduce another element on notation that will be convenient in the rest of this section.
\begin{dfn}
Fix integers $n,d$ and prime $p$. Let $f\in \F[\bx]$ be an
$n$-variate polynomial of degree at most $d$. We define
$LDT(f,n,d,p)$ to be $1$ if $\deg_p(f) =  d$, and $0$ otherwise.
\end{dfn}

Before proceeding, we note that the results of Section
\ref{sec:ldtest} imply that given $n,d,p$ and black-box access to
$f$, $LDT(f,n,d,p)$ can be computed in time $\poly(n)\cdot
O(p^{\left\lceil {d}/{(p-1)}\right\rceil +1})$. In particular, if
given $\ba\in \F^n$, we can compute $f(\ba)$ in $\poly(n)$-time,
then we can compute $LDT(f,n,d,p)$ in time
$\poly(n)\cdot O(p^{\left\lceil {d}/({p-1})\right\rceil +1})$.\\
The following lemma is an easy corollary of Lemma \ref{lem:degiffr-mon}.
\begin{lemma}
Fix integers $r,k$ with $r\leq k$. Suppose $p=r+1$ is prime. Let
$g\in \F[\bx]$ be homogenous of degree $k$ and computable in
$\poly(n)$-time. There is a randomized algorithm running in time
$$poly(n)\cdot O((r+1)^{\left\lceil \frac{k}{r}\right\rceil +1})$$
determining whether $g$ contains an $r$-monomial.
\end{lemma}

\begin{proof}

The algorithm simply returns $LDT(g,n,d=k,p=r+1)$.
The running time follows from the discussion above.
The correctness follows from Lemma \ref{lem:degiffr-mon}.
\end{proof}

We wish to have a similar result when $r+1$ is not a prime.

\begin{lemma}\label{lem:r,k,p}
Fix integers $r,k$ with $r\leq k$. Let $p$ be the smallest prime
such that $\frac{p-1}{r} \in \mathbb{Z}$. Let $g\in \F[\bx]$ be
homogenous of degree $k$ and computable by a $\poly(n)$-size
circuit. There is a randomized algorithm running in time
$poly(n)\cdot O(p^{\left\lceil \frac{k}{r}\right\rceil +1})$
determining whether $g$ contains an $r$-monomial.
\end{lemma}
\begin{proof}
Denote $l\triangleq \frac{p-1}{r}$ and define
\[h(x_1,x_2,\ldots ,x_n):= g(x^l_1,x^l_2,\ldots ,x^l_n).\]
The algorithm returns $LDT(h,n,d=k\cdot l,p)$.

Note that $h$ is homogenous of degree $k\cdot l$.
Note also that $h$ contains an $r\cdot l$-monomial if and only if $g$ contains
an $r$-monomial.
As $r\cdot l + 1=p$ correctness now follows from Lemma \ref{lem:degiffr-mon}.

\end{proof}

The best known bound for the smallest prime number $p$ that
satisfies $r|p-1$ is $r^{5.5}$ due to Heath-Brown \cite{R96}.
This gives a randomized algorithm running in time $$poly(n)\cdot
O(r^{\frac{5.5k}{r} +O(1)}).$$ Schinzel, Sierpinski, and Kanold
have conjectured the value to be 2 \cite{R96}. In the following
Theorem we give a better bound. We first give the following

\begin{lemma}\label{lem:r,k,p,l}
Fix integers $r,k$ with $r\leq k$. Let $p$ be the smallest prime
such that there is an $ l \in \mathbb{Z}$ for which $r\cdot l \leq
p-1 $ and $(r+1)\cdot l > p-1$. Let $g\in \F[\bx]$ be homogenous
of degree $k$ and computable by a $\poly(n)$-size circuit. There
is a randomized algorithm running in time
$$poly(n)\cdot O\left(p^{\left\lceil \frac{l\cdot k}{p-1}\right\rceil +1}\right)$$
determining whether $g$ contains an $r$-monomial.
\end{lemma}
\begin{proof}
As in the proof of Lemma \ref{lem:r,k,p}, we define
$h(x_1,x_2,...,x_n)\triangleq g(x^l_1,x^l_2,...,x^l_n)$. The
algorithm returns $LDT(h,n,d=k\cdot l,p)$. As in the proof of
Lemma $\ref{lem:r,k,p}$, $h$ is homogenous of degree $k\cdot l$
and contains an $(r\cdot l)$-monomial if and only if $g$ contains
an $r$-monomial. Furthermore, as $r\cdot l \leq p-1$ and
$(r+1)\cdot l \geq p$, $h$ contains a $(p-1)$-monomial if and only
if $g$ contains an $r$-monomial. Correctness now follows from
Lemma \ref{lem:degiffr-mon}.

\end{proof}

The main result of this section contains two results. The first is
unconditional. The second is true if Cramer's conjecture is
true. Cramer's conjecture states that the gap between two
consecutive primes $p_{n+1}-p_n=O(\log^2 p_n), \cite{C36}$.

\begin{theorem}
(Unconditional Result) Fix any integers $r,k$ with $2\leq r\leq
k$. Let $g\in \F[\bx]$ be homogenous of degree $k$ and computable
by a $\poly(n)$-size circuit. There is a randomized algorithm
running in time
 $$\poly(n)\cdot O\left(r^{\frac{2k}{r} +O(1)}\right)$$
 determining whether $g$ contains an $r$-monomial.

(Conditional Result) If Cramer's Conjecture is true then the time
complexity of the algorithm is
$$\poly(n)\cdot O\left( r^{\frac{k}{r}+o\left(\frac{k}{r}\right)}\right).$$
\end{theorem}

\begin{proof}
We will find $p$ and $l$ as required in Lemma \ref{lem:r,k,p,l}.
Fix a prime $p$ such that $r^2 + r +1 <p < 2r^2 + 2r\leq 3r^2 $ .
(This can be done as for any positive integer $t>3$, there is
always a prime between $t$ and $2t$.)

 Define $l\triangleq \lfloor \frac{p-1}{r} \rfloor $.
 We have
 \begin{eqnarray*}
 r\cdot l = r\cdot \lfloor \frac{p-1}{r}\rfloor \leq p-1
 \end{eqnarray*}
 \begin{eqnarray*}(r+1)\cdot l &\geq& (r+1)\cdot (\frac{p-1}{r} -1)\\
 &=&(p-1)+\frac{p-1}{r}-r-1 >(p-1)
 \end{eqnarray*}
The first claim now follows from Lemma \ref{lem:r,k,p,l} and Corollary
\ref{cor:enough_over_Fp}.

If Cramer's conjecture is true then there is a constant $c$ such
that for every integer $x$ there is a prime number in
$[x,x+c\log^2(x)]$ . Then there is a prime number $p$ in the
interval $[2cr\log^2 r, 2cr\log^2 r+c\log^2(2cr\log^2 r)]$ and we
can choose $l=2c\log^2 r$. Then the time complexity will be
$$poly(n)\cdot O\left( r^{\frac{k}{r}+o\left(\frac{k}{r}\right)}\right).$$
\end{proof}

\ignore{\begin{center}
\begin{tabular}{|l|l|l|l|}
Reference & Comment & $y=O(\cdot)$ & Result\\
\hline \hline \cite{C36} & Cramer's conjecture & $\log^2x$ & \\
\hline \cite{C36} & Riemann hypothesis & $\sqrt{x}\log x$& \\
\hline \cite{BHP01} & Unconditionally  & $x^{0.525}$ & \\
\hline
\end{tabular}
\end{center}

N. A. Carella and independently K. E. Lumbard claim that they
prove $y=\sqrt{x}\cdot poly(\log x)$ without any conditions in
unreviewed papers.}

The following table summarizes the result for $r\le 11$. See
Lemma~\ref{lem:r,k,p,l}.

\begin{center}
\begin{tabular}{|c|c|l|}
$r$ & Result & Field and $l$\\
\hline $1$ & $2^k$ \cite{W09} & $\FF_2$ , $l=1$\\
\hline $2$ & $1.73^k$
& $\FF_3$ , $l=1$  \\
\hline $3$ & $1.912^k$
& $\FF_7$ , $l=2$  \\
\hline $4$ & $1.495^k$
& $\FF_5$ , $l=1$  \\
\hline $5$ & $1.615^k$
& $\FF_{11}$ , $l=2$  \\
\hline $6$ & $1.383^k$
& $\FF_{7}$ , $l=1$  \\
\hline $7$ & $1.533^k$
& $\FF_{23}$ , $l=3$  \\
\hline $8$ & $1.424^k$
& $\FF_{17}$ , $l=2$ \\
\hline $9$ & $1.387^k$
& $\FF_{19}$ , $l=2$ \\
\hline $10$ & $1.27^k$
& $\FF_{11}$ , $l=1$\\
\hline $11$ & $1.329^k$
& $\FF_{23}$ , $l=2$\\
\hline
\end{tabular}
\end{center}

\end{document}